\documentclass[11pt]{article}
\usepackage{fullpage}
\usepackage{amsmath,amsfonts,amsthm,mathrsfs,mathpazo,xspace,hyperref,graphicx}
\usepackage{endnotes}
\usepackage{color}
\usepackage{bbm}
\usepackage{times}
\usepackage{amssymb,latexsym}

\usepackage{dsfont}
\usepackage{tikz}
\usepackage{hyperref}
\usepackage{blkarray, mathtools}
\usepackage{float}
\usepackage{upgreek, bbold, float}
\usepackage{caption}
\usepackage[boxed, ruled, vlined, linesnumbered]{algorithm2e}
\usetikzlibrary{plotmarks}

\newcommand{\footremember}[2]{%
    \footnote{#2}
    \newcounter{#1}
    \setcounter{#1}{\value{footnote}}%
}

\newtheorem{theorem}{Theorem}
\newtheorem{lemma}[theorem]{Lemma}
\newtheorem{proposition}[theorem]{Proposition}
\newtheorem{corollary}[theorem]{Corollary}

\theoremstyle{definition}
\newtheorem{definition}[theorem]{Definition}

\newcounter{example}

\newenvironment{example*}{\par\medskip\noindent\textbf{Example.}}{\par\hfill$\triangle$\par\medskip}

\begin{document}

\title{Fully device-independent quantum key distribution using synchronous correlations}

\author{Nishant Rodrigues
    \footremember{quics}{Joint Center for Quantum Information and Computer Science, College Park, MD} 
    \footremember{umdcs}{Department of Computer Science, University of Maryland, College Park, MD}
    \and Brad Lackey 
    \footremember{msrquantum}{Quantum Systems Group, Microsoft Quantum, Redmond, WA}
}

\maketitle

\newcommand{\bra}[1]{\ensuremath \langle{#1}|}%
\newcommand{\ket}[1]{{\ensuremath |{#1}\rangle}}%
\newcommand{\bracket}[2]{\ensuremath \langle{#1}|{#2}\rangle}%
\newcommand{\h}[1]{\ensuremath \mathfrak{#1}}%
\newcommand{\indicator}[1]{\ensuremath \mathbb{1}_{\{#1\}}}%
\newcommand{\id}{\ensuremath \mathrm{id}}%
\newcommand{\tr}{\ensuremath{\mathrm{tr}}}
\newcommand{\ip}[2]{\ensuremath \langle{#1}|{#2}\rangle}%
\newcommand{\category}[1]{\ensuremath{\mathsf{#1}}}
\newcommand{\Hom}{\mathrm{Hom}}
\newcommand{\given}{\ensuremath \:|\:}
\newcommand{\covec}[1]{\ensuremath \undertilde{#1}}

\begin{abstract}
We derive a device-independent quantum key distribution protocol based on synchronous correlations and their Bell inequalities. This protocol offers several advantages over other device-independent schemes including symmetry between the two users and no need for preshared randomness. We close a ``synchronicity'' loophole by showing that an almost synchronous correlation inherits the self-testing property of the associated synchronous correlation. We also pose a new security assumption that closes the ``locality'' (or ``causality'') loophole: an unbounded adversary with even a small uncertainty about the users' choice of measurement bases cannot produce any almost synchronous correlation that approximately maximally violates a synchronous Bell inequality.
\end{abstract}

\begin{section}{Introduction}

Quantum key distribution (QKD) allows two parties to establish a shared classical secret key using quantum resources. Two main requirements of QKD are
\begin{enumerate}
    \item Correctness: the two parties, Alice and Bob, get the same key; and
    \item Security: an adversary Eve gets negligible information about the key.
\end{enumerate}
Device-independent quantum key distribution (DI-QKD) is entanglement-based, and aims to prove security of QKD based solely on the correctness of quantum mechanics, separation of devices used by the two parties, and passing of statistical tests known as Bell violations \cite{vazirani2014fully,miller2016robust}. These protocols are usually specified by a non-local game, characterized by a conditional probability distribution or \emph{correlation} $p(y_A, y_B \given x_A, x_B)$. Intuitively, Alice and Bob obtain or generate random inputs $x_A$ and $x_B$ respectively, and the correlation describes the likelihood their entangled quantum devices return outputs $y_A$ and $y_B$ to each respectively. We will be interested in symmetric correlations and so will take $x_A, x_B \in X$ and $y_A, y_B \in Y$ where $X$ and $Y$ are finite sets; for our protocol specifically $X = \{0,1,2\}$ and $Y = \{0,1\}$.

In general, security of a DI-QKD scheme relies on the concept of monogamy of entanglement. The key mathematical result is that maximally entangled quantum states are separable within any larger quantum system. In cryptographic terms, if Alice and Bob share a maximally entangled state then the results of measurements they make on this state will be uncorrelated to any other measurement results an adversary can perform. Hence, presuming the correctness of quantum mechanics, no adversary can have any information about key bits Alice and Bob may generate through this process. At a high level, a DI-QKD protocol will involve two types of rounds: testing rounds where Alice and Bob (publicly) share their inputs and output results for performing statistics tests, and data rounds where they obtain shared secret bits. The goal of the testing rounds is to produce a certificate that Alice and Bob are operating on maximally entangled states.

Most current DI-QKD schemes are based on the CHSH inequality. This is a linear inequality in the correlation terms $p(y_A, y_B \given x_A, x_B)$, which if satisfied characterizes classical statistics within a quantum system. Hence a violation of this inequality is a certificate of quantum behavior. This inequality exhibits ``rigidity'' in that the only quantum state that produces a maximal violation of the inequality is (up to natural equivalences) a Bell pair: two maximally entangled qubits. Thus the goal of the testing rounds in a DI-QKD protocol is to statistically verify that the system produces a maximal violation of the CHSH inequality.

A technical assumption made on a non-local game is, that while Alice and Bob may preshare an entangled resource in each round, they are not allowed any communication between receiving or generating their inputs $x_A$ and $x_B$ and measuring the system to obtain their outputs $y_A$ and $y_B$. This is typically called a ``no-signaling'' condition, leading to \emph{nonsignaling correlations} which include all quantum strategies. If (classical) communication between Alice and Bob is possible, then it is simple to classically simulate a correlation that produces a maximal violation of the CHSH inequality, and hence any certificates of quantumness or entanglement are void \cite{toner2003communication}. This \emph{locality} or \emph{causality loophole} in the security proof is challenging to avoid; the only known means to close it is by having Alice and Bob acausally separated during each round: bounds on the speed of light prevent such communication \cite{hensen2015loophole,giustina2015significant,shalm2015strong}. 

A \emph{synchronous correlation} is one such that $p(y_A, y_B \given x, x) = 0$ whenever $y_A \neq y_B$ and $x \in X$. That is, whenever Alice and Bob input the same value they are guaranteed to receive the same outputs, although that value may be nondeterministic. These correlations have recently become popular owing to their use in the resolution of the Connes Embedding Conjecture and Tsirl'son's Problem \cite{ji2020mip}, but have also been used to generalize combinatorial properties to the quantum setting \cite{mancinska2014graph,paulsen2016estimating,kim2018synchronous}. While it is far from obvious, every synchronous quantum correlation is symmetric (see Appendix \ref{appendix:synchronous} for details).

We present a fully device-independent QKD protocol based on synchronous correlations. This protocol has the property of symmetry between Alice and Bob, by which we mean that their roles are completely interchangeable. This is an advantage over other DI-QKD protocols based on the CHSH inequality \cite{vazirani2014fully} (which is neither symmetric nor synchronous) as sender versus receiver roles do not need to be negotiated. Additionally, as Alice and Bob select their inputs independently they do not need preshared secret bits to decide when to perform testing versus data rounds.

The mathematical framework needed to prove device-independent security of this protocol was laid out in \cite{rodrigues2017nonlocal}: four analogues of the Bell/CHSH inequality for synchronous correlations were given (in this work we focus only on one of these $J_3(p) \geq 0$, see (\ref{eqn:effective-J3}) below for its explicit form), bounds on quantum violations of these--so call Tsirl'son bounds--were characterized ($J_3(p) \geq -\frac{1}{8}$), and rigidity of correlations that achieve a maximal violation by Bell states proven. The two critical analyses needed to complete a proof of security for our DI-QKD protocol are as follows.
\begin{itemize}
    \item As one cannot statistically guarantee maximal violation of a Bell inequality, we must prove that if the system is observed to be close to the maximal violation then it is close to the ideal system involving measuring a Bell pair.
    \item Provide an alternative security assumption that bypasses the causality loophole as described above.
\end{itemize}

We tackle the first of these through two theorems. For context, Alice and Bob will independently uniformly select their inputs from $X = \{0,1,2\}$ and each measure a quantum system that produces a bit for output $Y = \{0,1\}$. The ideal system, that produces $J_3(p) = -\frac{1}{8}$, involves measuring a Bell pair using three specific projection-valued measures $\{\hat{E}^x_y\}_{y=0,1}$ for $x = 0,1,2$; these are given in (\ref{eqn:projections-exy}) below. In Section \ref{sec:protocol} we show that if we take a \emph{synchronous} quantum system that is close to achieving maximal $J_3$ violation, then it must be close to the ideal system. Our initial Protocol A, given as Algorithm \ref{alg:synchronousQKD} below, captures this simplified DI-QKD method.


\begin{theorem}{(Informal)}
Let $p(y_A, y_B \given x_A, x_B) = \frac{1}{d}\tr(E^{x_A}_{y_A}E^{x_B}_{y_B})$ be a synchronous quantum correlation, where for each $x=0,1,2$ we have $\{E^x_y\}_{y=0,1}$ are projection-valued measures on a $d$-dimensional Hilbert space $\mathfrak{H}$. Suppose $J_3(p) \leq -\frac{1}{8} + \lambda$. Then each $E^x_y \approx \tilde{E}^x_y = L^x_y + \hat{E}^x_y\otimes \mathbb{1}$, in that there is a universal constant $C$ with
$$\frac{1}{3}\sum_{x,y} \frac{1}{d}\tr\left( \left( E^x_y - \tilde{E}^x_y\right)^2\right) \leq C \lambda,$$
where $L^x_y$ are projections on a subspace of dimension at most $\lambda d$.
\end{theorem}

Unfortunately this result introduces a new ``synchronicity loophole'' in our prove of security: rigidity holds among synchronous correlations, but are there nonsynchronous correlations with $J_3 = -\frac{1}{8}$ that cannot be use to certify maximal entanglement? In Section \ref{section:measures}, we close this loophole using recent work on ``almost synchronous'' correlations \cite{vidick2021almost}. This leads to our complete DI-QKD scheme Protocol B, given as Algorithm \ref{alg:synchronousQKD-protocolB} below. Extended to our framework the result is informally stated as follows.

\begin{theorem}{(Informal)}
Let $p(y_A, y_B \given x_A, x_B) = \bra\psi E^{x_A}_{y_A}\otimes (E^{x_B}_{y_B})^T \ket\psi$ be a symmetric correlation with asynchronicity $S = \frac{1}{3}\sum_x \sum_{y_A \not= y_B} p(y_A, y_B \given x, x)$, and let $\ket\psi = \sum_{j=1}^r \sqrt{\sigma_j} \sum_{m=1}^{d_j} \ket{\phi_{j,m}}\otimes \ket{\phi_{j,m}}$ be the Schmidt decomposition. Suppose $J_3(p) \leq -\frac{1}{8} + \lambda$. Then there exist synchronous correlations with projection-values measures $\tilde{E}^{j,x}_y = L^{j,x}_y + \hat{E}^x_y\otimes \mathbb{1}$ defined on Hilbert spaces $\mathfrak{H}_j = \mathrm{span}\{\ket{\phi_{j,m}}\}$ where $p$ is close to the convex sum of these, in that there are universal constants $c,C_1,C_2$ where
$$\frac{1}{3}\sum_{x,y} \sum_{j=1}^{r} \sigma_j d_j \left(\frac{1}{d_j}\sum_{m=1}^{d_j}\bra{\phi^A_{j,m}}(E^{x}_y - \tilde{E}^{j,x}_y)^2\ket{\phi^A_{j,m}}\right) \leq C_1S^{c} + C_2\lambda.$$
\end{theorem}

Finally, in Section \ref{section:causality-loophole}, we pose a new security assumption to close the \emph{causality} or \emph{locality} loophole. Our new security assumption gives unlimited communication and computational power to the adversary Eve, but assumes that she has imperfect knowledge of Alice and Bob's inputs. For Eve's uncertainty about Alice and Bob's inputs, denoted by $\epsilon$, where $0 \leq \epsilon \leq 2/3$, we derive information theoretic bounds for how much that uncertainty is allowed to grow before she ends up with an infeasible cheating strategy. We state our result informally as follows.

\begin{theorem}{(Informal)}
Let $0 \leq \lambda \leq \frac{1}{8}$ and $0\leq \mu \leq \mu_0$ be allowed errors in Alice and Bob's Bell term $J_3$ and asynchronicity $S$ respectively. Also, let $\tilde{J}_3$ and $\tilde{S}$ be analogous Bell inequality and asynchronicity terms for Eve's strategy. 
For $0\leq \delta \leq \tilde{S}$, there exists a function $f: \mathbb{R} \times \mathbb{R} \times \mathbb{R} \to \mathbb{R}$ of $\delta, \mu$ and $\lambda$ such that if $\epsilon^{\delta}_{max} = 2/3 - f(\delta,\mu,\lambda)$, and
Eve's uncertainty is $\epsilon > \epsilon^{\delta}_{max}$ then every correlation satisfies  $\tilde{S} < \delta$, and hence there is no feasible strategy she can produce.
\end{theorem} 

The conclusion is that Eve must have close to perfect knowledge of Alice and Bob's inputs to successfully simulate the statistics for the protocol. We derive an expression for $f$ in the theorem above, and plot Eve's uncertainty against varying values of Alice and Bob's allowed asynchronicity. 

In this paper we focus on the tools needed to complete the security proof, including mathematical analysis of the asynchronous case, and the analysis of the new security assumption to resolve the causality loophole. The security proof for the protocol then follows from arguments in \cite{arnon2019simple} where they present a framework based on the Entropy Accumulation Theorem of \cite{dupuis2020entropy} to analyze security and correctness of device-independent QKD protocols. 

\end{section}

\begin{section}{Preliminaries}
    We present some definitions that will be used in the protocol later. Like other device-independent schemes, our protocol is expressed in terms of a nonlocal game, which is characterized by a conditional probability distribution (or correlation) $p(y_A, y_B | x_A, x_B)$ where $x_A, x_B \in X$, and $y_A, y_B \in Y$ for finite sets $X$ and $Y$. By a nonlocal game we mean the players Alice and Bob will receive inputs $x_A, x_B \in X$ from a referee and will produce outputs $y_A, y_B \in Y$. These are then adjudicated by the referee against some criterion, which we will discuss below. Alice and Bob are allowed to use pre-shared information (such as classical randomness or entangled states), however, they are not allowed to communicate once they receive their inputs \cite{cleve2004consequences}. This is characterized by the famous nonsignaling conditions \cite{popescu1994quantum}, which for completeness we express here.
    
    \begin{definition}
        A correlation $p$ is \emph{nonsignaling} if it satisfies (i) for all $y_A,x_A,x_B,x_B'$
        $$\sum_{y_B} p(y_A,y_B\:|\:x_A,x_B) = \sum_{y_B} p(y_A,y_B\:|\:x_A,x_B'),$$
        and (ii) for all $y_B,x_B,x_A,x_A'$
        $$\sum_{y_A} p(y_A,y_B\:|\:x_A,x_B) = \sum_{y_A} p(y_A,y_B\:|\:x_A',x_B).$$
    \end{definition}
    
    We have selected our notation above so as to emphasize a symmetry between Alice and Bob. This will be a consequence of our correlations being \textit{synchronous}, which will form the basis of our quantum key distribution protocol. Formally:
    \begin{definition}
        A correlation is \emph{synchronous} if
        \begin{equation}\label{eqn:syncrhonous:definition}
           p(y_A,y_B\:|\:x,x) = 0 \text{ if $x\in X$ and $y_A\not= y_B$ $\in Y$.} 
        \end{equation}
    A correlation is \emph{symmetric} if $p(y_A,y_B\:|\:x_A,x_B) = p(y_B,y_A\:|\:x_B,x_A)$.
    \end{definition}
    It is straightforward for Alice and Bob to create a nonlocal game with synchronous correlation, regardless of how the referee selects $x_A, x_B \in X$: they agree on some function $f:X \to Y$ and output $y_A = f(x_A)$ and $y_B = f(x_B)$. The value of a nonlocal game is the expected success probability that Alice and Bob produce outputs of the desire form; since this is always $1$ for synchronous games the value plays little role.
    
    As is traditional with schemes derived from the CHSH or Magic Square games, or their generalizations \cite{mermin1990simple,peres1990incompatible,cleve2004consequences,arkhipov2012extending,coladangelo2017robust}, the analysis relies on understanding the space of local (or ``classical'' or ``hidden variables'') correlations. To set notation, we use $\h{H}_A, \h{H}_B$ to denote finite-dimensional Hilbert spaces. For each input $x \in X$, let $\{E^x_y\}_{y \in Y}$ be a POVM with measurement outcomes $y\in Y$, that is each $E^x_y$ is a positive operator and $\sum_y E^x_y = \mathbb{1}$. A projection-valued measure is a POVM $\{E^x_y\}_{y \in Y}$ where each $E^x_y$ is a projection. Then classical and quantum correlations are formally defined as follows.
    
    \begin{definition}
        A \emph{local hidden variables strategy}, or simply \emph{classical correlation}, is a correlation of the form
        \begin{equation*}
            p(y_A,y_B\:|\:x_A,x_B) = \sum_{\omega\in\Omega} \mu(\omega) p_A(y_A\:|\: x_A,\omega) p_B(y_B\:|\: x_B,\omega)
        \end{equation*}
        for some finite set $\Omega$ and probability distribution $\mu$.
        A \emph{quantum correlation} is a one that takes the form
        \begin{equation*}
            p(y_A,y_B\:|\:x_A,x_B) = \tr(\rho(E^{x_A}_{y_A}\otimes F^{x_B}_{y_B}))
        \end{equation*}
        where $\rho$ is a density operator on the Hilbert space $\h{H}_A\otimes\h{H}_B$, and for each $x\in X$ we have $\{E^x_y\}_{y\in Y}$ and $\{F^x_y\}_{y\in Y}$ are POVMs on $\h{H}_A$ and $\h{H}_B$ respectively. 
    \end{definition}
    
    Synchronous classical and quantum correlations can be further characterized. For example, every synchronous classical correlation arises from a generalization of the simple strategy above: Alice and Bob (randomly) pre-select a function $f:X \to Y$ and upon receiving $x_A,x_B \in X$ (deterministically) compute their outputs $y_A = f(x_A)$ and $y_B = f(x_B)$. Similarly, every synchronous quantum correlation can be expressed as a convex combination of so-called ``tracial'' states on projection-valued measures \cite{paulsen2016estimating,rodrigues2017nonlocal}
    $$p(y_A,y_B\:|\: x_A,x_B) = \frac{1}{d}\tr(E^{x_A}_{y_A}E^{x_B}_{y_B}).$$
    
    For input and output $X = \{0,1,2\}$ and $Y = \{0,1\}$, respectively, there are four Bell inequalities for synchronous hidden variables theories. By this we mean that the synchronous classical correlations (among general non-signalling synchronous correlations) are characterized by four inequalities $J_0, J_1, J_2, J_3 \geq 0$ where each $J_i = J_i(p)$ is a linear combination of the correlation components $p(y_A,y_B\:|\: x_A,x_B)$. See (\ref{eqn:Bell-inequalities}) in Appendix \ref{appendix:synchronous} for explicit formulas. For this work, we will focus only on one of these as given in (\ref{eqn:effective-J3}) below. 
    
    Synchronous quantum correlations can violate the inequality $J_3 \geq 0$. However one can show an analogue of Tsirl'son bound, in that any synchronous quantum correlation must have $J_3 \geq -\frac{1}{8}$. Of particular interest are correlations that maximize this quantum violation. Like CHSH or Magic Square games, one can show a rigidity result: there is a unique synchronous quantum correlation with $J_3 = -\frac{1}{8}$, which must involve a maximally entangled state shared between Alice and Bob. One can then use principle decomposition, or two projections theory, to convert this into a self-test for certifying a single EPR pair, hence the basis for device-independence. See Appendix \ref{appendix:synchronous} for a more detailed discussion of these results.

\end{section}

\begin{section}{A synchronous DI-QKD protocol} \label{sec:protocol}

    Here we present an initial form for a synchronous device-independent quantum key distribution protocol, Algorithm \ref{alg:synchronousQKD} below. The protocol uses a non-local game based on synchronous quantum correlations. To achieve device independence, we use a rigidity result based on Tsirl'son bounds for synchronous quantum correlations. We discuss these correlations in detail and state the Tsirl'son bounds in Appendix \ref{appendix:synchronous}. This protocol is symmetric with respect to Alice and Bob, each performs exactly the same task.
    
    Suppose Alice and Bob share an EPR pair. Each independently draws a uniformly random input $x_A, x_B \in X = \{0,1,2\}$ respectively, and measures according to $\{\hat{E}^{x_A}_y\}_{y \in Y}$ and $\{\hat{E}^{x_B}_y\}_{y \in Y}$ to get outputs $y_A, y_B \in Y = \{0,1\}$, where the projection-valued measures $\{\hat{E}^x_y\}_{y \in \{0,1\}}$ for $x \in \{0,1,2\}$ are defined as follows:
    \begin{equation}\label{eqn:projections-exy}
        \begin{aligned}
            \hat{E}^0_1 &= \ket{\phi_0}\bra{\phi_0}\text{, } \hat{E}^0_0 = \mathbb{1} - \hat{E}^0_1 \text{, } \quad \text{where } \ket{\phi_0} = \ket{1}\\ 
            \hat{E}^1_1 &= \ket{\phi_1}\bra{\phi_1}\text{, } \hat{E}^1_0 = \mathbb{1} - \hat{E}^1_1 \text{, } \quad \text{where } \ket{\phi_1} = \frac{\sqrt{3}}{2}\ket{0} + \frac{1}{2}\ket{1}\\
            \hat{E}^2_1 &= \ket{\phi_2}\bra{\phi_2}\text{, } \hat{E}^2_0 = \mathbb{1} - \hat{E}^2_1 \text{, } \quad \text{where } \ket{\phi_2} = \frac{\sqrt{3}}{2}\ket{0} - \frac{1}{2}\ket{1}\\
        \end{aligned}
    \end{equation}
    Alice's and Bob's results are characterized by the correlation
     \begin{equation*}
        p(y_A, y_B\:|\:x_A, x_B) = \frac{1}{2} \tr(\hat{E}^{x_A}_{y_A} \hat{E}^{x_B}_{y_B})
    \end{equation*}
    (see Theorem \ref{theorem:tracial-state} in Appendix \ref{appendix:synchronous}). In particular, Alice and Bob's strategy produces the synchronous quantum correlation with correlation matrix:
    \begin{equation}\label{eqn:protocol-correlation}
        [p(y_A,y_B|x_A,x_B)] = \frac{1}{8} \quad \begin{blockarray}{cccccccccc}
            (0,0) & (0,1) & (0,2) & (1,0) & (1,1) & (1,2) & (2,0) & (2,1) & (2,2) \\
            \begin{block}{(ccccccccc)@{\quad }l} 
                4 & 1 & 1 & 1 & 4 & 1 & 1 & 1 & 4 & (0,0)\\
                0 & 3 & 3 & 3 & 0 & 3 & 3 & 3 & 0 & (0,1)\\
                0 & 3 & 3 & 3 & 0 & 3 & 3 & 3 & 0 & (1,0)\\
                4 & 1 & 1 & 1 & 4 & 1 & 1 & 1 & 4 & (1,1)\\
            \end{block}
        \end{blockarray}
    \end{equation}
    
    One can easily verify this correlation yields a maximal violation of the Bell inequality, $J_3 = -\frac{1}{8}$, where
    \begin{equation}\label{eqn:effective-J3}
        \begin{array}{rcr@{\:}l}
            J_3 &=& 1 - \frac{1}{4} \big( & p(0,1\:|\:0,1) + p(1,0\:|\:0,1) + p(0,1\:|\:1,0) + p(1,0\:|\:1,0)\\
            && +& p(0,1\:|\:0,2) + p(1,0\:|\:0,2) + p(0,1\:|\:2,0) + p(1,0\:|\:2,0)\\
            && +& p(0,1\:|\:1,2) + p(1,0\:|\:1,2) + p(0,1\:|\:2,1) + p(1,0\:|\:2,1)\:\big).
        \end{array}
    \end{equation}
    In Appendix \ref{appendix:synchronous} we discuss the rigidity of this correlation. Specifically, one has that any synchronous quantum correlation that achieves $J_3 = -\frac{1}{8}$ must have implemented the strategy above. That is, this maximal violation of $J_3$ is a self-test of the device to detect interference from adversary: Alice and Bob can certify that their devices hold maximally entangled pairs, and by monogamy of entanglement can establish that Eve doesn't have any information about their inputs.
    
    Our initial protocol extends the above scenario to $n$ rounds. It is important to note that the observable for our synchronous Bell inequality (\ref{eqn:effective-J3}) only involves correlations where Alice and Bob use different inputs. This leads to two significant theoretical advantages of our system.
    \begin{enumerate}
        \item Neither Alice nor Bob must pre-select which rounds will used for testing versus key generation. Upon revealing their choices of bases, testing rounds given by those where they selected different bases and key generation rounds where they selected the same basis. In particular, they need not have any pre-shared randomness.
        \item Every round is effective, in that every testing round improves the estimate of $J_3$ and every key generation round produces one bit of uniform shared secret. 
    \end{enumerate}
    Of course no physical device adheres to theoretical model perfectly, so in practice one still must perform standard information reconciliation and privacy amplification on the results. The full protocol is presented in Algorithm \ref{alg:synchronousQKD}.
    
    Once the $n$ rounds of the protocol are over, Alice and Bob communicate their basis selection over an authenticated classical channel. In the case that they chose different bases (i.e. $x_A \neq x_B$), they exchange their measurement outcomes and use those to compute $J_3$. If the value of $J_3$ deviates too much from $-\frac{1}{8}$, they abort. The protocol is synchronous, therefore $y_A = y_B$ whenever $x_A = x_B$ and those can be used as the raw key bits for further standard privacy amplification and information reconciliation. 
    
    \begin{table}[t]
    \begin{algorithm}[H]\label{alg:synchronousQKD}
        \caption{Protocol A}
        \SetAlgoLined
        \DontPrintSemicolon
        \SetKwFor{For}{for}{do}{}
        \KwIn{$\lambda, n$}
        $X \leftarrow \{0,1,2\}$ and $Y \leftarrow \{0,1\}$\\
        Alice and Bob share $n$ EPR pairs: $\ket{\psi} = \frac{1}{\sqrt{2}}\left(\ket{00} + \ket{11}\right)$\\
        \For{$i = 1, \cdots, n$}{
            Alice draws $x^i_A \overset{\$}{\leftarrow} X$ and Bob draws $x^i_B \overset{\$}{\leftarrow} X$\\
            With the $i^{\text{th}}$ EPR pair, Alice obtains $y^i_A$ using $\{E^{x^i_A}_y\}$ and Bob obtains $y^i_B$ using $\{E^{x^i_B}_y\}$
        }
        Alice and Bob exchange their choices of $x^i_A, x^i_B$, for $i \in [n]$\\
        Whenever $x^i_A \neq x^i_B$, Alice and Bob exchange their results $y^i_A, y^i_B$\\
        $k \leftarrow \emptyset$\\
        \For {$i = 1,\cdots, n$}{
            \uIf{$x^i_A = x^i_B$}{
                $k \leftarrow k \cup y^i$, where $y^i := y^i_A = y^i_B$ due to synchronicity
            }
            \lElse{
                Add result to estimation of $J_3$
            }            
        }
        Compute an estimate $\hat{J}_3$ using (\ref{eqn:effective-J3})\\
        \uIf{ $|\hat{J}_3 + \frac{1}{8}| \leq \lambda$}{
            \textbf{Return } $k$ (for standard information reconciliation and privacy amplification)
        }\lElse{ Abort }
    \end{algorithm}
    \end{table}

    Our first main result is that the rigidity of synchronous quantum correlations with $J_3 = -\frac{1}{8}$ does show that nearby synchronous quantum correlations have the desired security.

    \begin{theorem}\label{theorem:main_bound}
        Let $p(y_A, y_B\:|\:x_A, x_B) = \frac{1}{d} \tr(E^{x_A}_{y_A} E^{x_B}_{y_B})$ be a synchronous quantum correlation with maximally entangled state, where $\{E^x_y\}$ is a projection-valued measure on a $d$-dimensional Hilbert space $\mathfrak{H}$. Suppose $J_3(p) \leq -\frac{1}{8} + \lambda$. Then on $\mathfrak{H} = \mathfrak{L} \oplus (\mathbb{C}^2\otimes\mathfrak{K})$ there exists a projection-value measure $\{\tilde{E}^x_y\}$ where
        \begin{enumerate}
            \item $\tilde{E}^x_y = L^x_y + \hat{E}^x_y\otimes \mathbb{1}_{\mathfrak{K}}$,
            \item $\frac{\dim{\mathfrak{L}}}{\dim{\mathfrak{H}}} \leq 8\lambda$,
            \item $\frac{1}{3} \sum_{x,y} \frac{1}{d}\tr\left(\left(E^x_y - \tilde{E}^x_y\right)^2\right) \leq 8\lambda$.
        \end{enumerate}
        In particular, the expected statistical difference
        $$\frac{1}{3}\sum_{x,y}\left|p(y,y\:|\:x,x) - \frac{1}{2}\right| \leq \sqrt{8}\sqrt{\lambda} + \tfrac{64}{3}\lambda.$$
    \end{theorem}
\begin{proof}
    We begin by defining the $\pm 1$-valued observables $M_x = E^x_0 - E^x_1$, so $M_x^2 = \mathbb{1}$, and following customary notation write
    $$a_x = \frac{1}{d}\tr(M_x) \text{ and } c_{x_Ax_B} = \frac{1}{d}\tr(M_{x_A}M_{x_B}).$$
    Similarly denote $\tilde{M}_x = \tilde{E}^x_0 - \tilde{E}^x_1$. Notice $E^x_0 = \frac{1}{2}(\mathbb{1} + M_x)$ and $E^x_1 = \frac{1}{2}(\mathbb{1} - M_x)$ so
    $$\frac{1}{3} \sum_{x,y} \frac{1}{d}\tr\left(\left(E^x_y - \tilde{E}^x_y\right)^2\right) = \frac{1}{6}\sum_x \frac{1}{d}\tr\left(\left(M_x - \tilde{M}_x\right)^2\right).$$
    
    Now define $\Delta := M_0 + M_1 + M_2$, and compute
    \begin{align}
        \Delta^2 &= M_0^2 + M_1^2 + M_2^2 + M_0M_1 + M_1M_0 + M_0M_2 + M_2M_0 + M_1M_2 + M_2M_1 \notag\\
        &= 3\mathbb{1} + M_0M_1 + M_1M_0 + (M_0 + M_1)M_2 + M_2(M_0 + M_1) \label{eqn:Delta-squared}\\
        &= \mathbb{1} + M_0M_1 + M_1M_0 + (M_0 + M_1 + M_2)M_2 + M_2(M_0 + M_1 + M_2) \notag\\
        &= \mathbb{1} + M_0M_1 + M_1M_0 + \Delta M_2 + M_2\Delta \label{eqn:Delta-squared-reduced}
    \end{align}
    We have $\Delta^2$ relates to $J_3$, and hence we obtain the following bound:
    \begin{align}
        \frac{1}{d}\tr(\Delta^2) &= \frac{1}{d}\tr\left(M_0^2 + M_1^2 + M_2^2 + 2M_0M_1 + 2M_0M_2 + 2M_1M_2\right) \notag\\
        &= \frac{3}{d}\tr\left(\mathbb{1}\right) + \frac{2}{d}\tr\left(M_0M_1 + M_0M_2 + M_1M_2\right) \notag\\
        &= 3 + 2(c_{01} + c_{02} + c_{12}) = 1 + 2(1 + c_{01} + c_{02} + c_{12}) = 1 + 8J_3 \notag\\
        &\leq 1 + 8\left(-\frac{1}{8} + \lambda\right) = 8\lambda \label{eqn:trace-Delta-squared-bound}
    \end{align}
        
    Using two projections theory \cite{amrein1994pairs,halmos1969two,boettcher2010gentle}, we have a decomposition of the Hilbert space $\h{H}$
    $$\h{H} = \h{L}_{00} \oplus \h{L}_{01} \oplus \h{L}_{10} \oplus \h{L}_{11} \oplus \bigoplus_{j=1}^k \h{H}_j,$$ where $\dim(\h{L}_{\alpha\beta}) = l_{\alpha\beta}$ for $\alpha,\beta \in \{0,1\}$, and $\dim(\h{H}_j) = 2$, where the projections $E^0_0$ and $E^1_0$ take the form:
    \begin{align*}
    E^0_0 &= 0_{l_{00}} \oplus 0_{l_{01}} \oplus \mathbb{1}_{l_{10}} \oplus \mathbb{1}_{l_{11}} \oplus \bigoplus_{j=1}^k \left(\begin{array}{cc} 1 & 0 \\ 0 & 0\end{array}\right)\\
    E^1_0 &= 0_{l_{00}} \oplus \mathbb{1}_{l_{01}} \oplus 0_{l_{10}} \oplus \mathbb{1}_{l_{11}} \oplus \bigoplus_{j=1}^k \left(\begin{array}{cc} \cos^2\theta_j & \sin\theta_j\cos\theta_j \\ \sin\theta_j\cos\theta_j & \sin^2\theta_j\end{array}\right).
    \end{align*}
    That is, we can express
    \begin{align*}
        M_0 &= -\mathbb{1}_{\h{L}_{00}} \oplus -\mathbb{1}_{\h{L}_{01}} \oplus \mathbb{1}_{\h{L}_{10}} \oplus \mathbb{1}_{\h{L}_{11}} \oplus \bigoplus_{j=1}^k \left(\begin{array}{cc} 1 & 0 \\ 0 & -1\end{array}\right),\\
        M_1 &= -\mathbb{1}_{\h{L}_{00}} \oplus \mathbb{1}_{\h{L}_{01}} \oplus -\mathbb{1}_{\h{L}_{10}} \oplus \mathbb{1}_{\h{L}_{11}} \oplus \bigoplus_{j=1}^k \left(\begin{array}{cc} \cos 2\theta_j & \sin 2\theta_j \\ \sin2\theta_j  & -\cos2\theta_j\end{array}\right).
    \end{align*}
    Now let us define $\tilde{M}_0, \tilde{M}_1, \tilde{M}_2$ as follows. Note that our ideal projections $\hat{E}^1_0, \hat{E}^1_1$ correspond to angle $\hat{\theta} = \frac{2\pi}{3}$, and without loss of generality we can assume\footnote{Direct examination of (\ref{eqn:projections-exy}) reveals that any $\theta_j$ is within $\frac{\pi}{6}$ of the image of some $E^x_y$; the bound we prove is symmetric in $x,y$ we may reorder the labeling in each $\mathfrak{H}_j$ so that $\theta_j$ is close to $E^1_0$ with $\hat\theta = \frac{2\pi}{3}$.} $|\theta_j - \hat{\theta}| \leq \frac{\pi}{6}$.  
    \begin{align*}
        \tilde{M}_0 = M_0 &= -\mathbb{1}_{\h{L}_{00}} \oplus -\mathbb{1}_{\h{L}_{01}} \oplus \mathbb{1}_{\h{L}_{10}} \oplus \mathbb{1}_{\h{L}_{11}} \oplus \bigoplus_{j=1}^k \left(\begin{array}{cc} 1 & 0 \\ 0 & -1\end{array}\right),\\
        \tilde{M}_1 &= -\mathbb{1}_{\h{L}_{00}} \oplus \mathbb{1}_{\h{L}_{01}} \oplus -\mathbb{1}_{\h{L}_{10}} \oplus \mathbb{1}_{\h{L}_{11}} \oplus \bigoplus_{j=1}^k \left(\begin{array}{cc} \cos 2\hat\theta & \sin 2\hat\theta \\ \sin2\hat\theta  & -\cos2\hat\theta\end{array}\right),\\
        \tilde{M}_2 &= \mathbb{1}_{\h{L}_{00}} \oplus \mathbb{1}_{\h{L}_{01}} \oplus -\mathbb{1}_{\h{L}_{10}} \oplus -\mathbb{1}_{\h{L}_{11}} \oplus \bigoplus_{j=1}^k \left(\begin{array}{cc} -1-\cos 2\hat\theta & -\sin 2\hat\theta \\ -\sin2\hat\theta  & 1+\cos2\hat\theta\end{array}\right).
    \end{align*}
    As desired, $\tilde{M}_x = (L^x_0 - L^x_1) + \hat{M}_x\otimes\mathbb{1}_{\mathbb{C}^k}$, where the $\{L^x_y\}$ are the projection onto the summands $\h{L}_{\mu\nu}$.
    
    
    First we bound the dimension of each $\mathfrak{L}_{\mu\nu}$. Consider the relation (\ref{eqn:Delta-squared}) for $\Delta^2$. If $\ket{\psi_{01}}\in \mathfrak{L}_{01}$, then
    \begin{align*}
        \bra{\psi_{01}} \Delta^2 \ket{\psi_{01}} &= \bra{\psi_{01}} (3\mathbb{1} + M_0M_1 + M_1M_0 + (M_0 + M_1)M_2 + M_2(M_0 + M_1) \ket{\psi_{01}}\\
        &= 3 - 1 - 1 + 0 + 0 = 1.
    \end{align*}
    The same equality holds for $\ket{\psi_{10}}\in \mathfrak{L}_{10}$, namely $\bra{\psi_{10}} \Delta^2 \ket{\psi_{10}} = 1$.
    
    For a vector $\ket{\psi_{00}}$ in $\h{L}_{00}$ we again use relation (\ref{eqn:Delta-squared}) to get 
    $$\bra{\psi_{00}} \Delta^2 \ket{\psi_{00}} =  3 + 1 + 1 - 4\bra{\psi_{00}} M_2 \ket{\psi_{00}}.$$
    Now from Cauchy-Schwarz, and that $M_2^2 = \mathbb{1}$, we have
    $$| \bra{\psi_{00}} M_2 \ket{\psi_{00}} | \leq | \langle {\psi_{00}} | {\psi_{00}} \rangle |^{\frac{1}{2}} | \bra{\psi_{00}} M_2^2 \ket{\psi_{00}} |^{\frac{1}{2}} = 1$$
    and thus $\bra{\psi_{00}} \Delta^2 \ket{\psi_{00}} \geq 1$. Similarly for  $\ket{\psi_{11}}$ in $\h{L}_{11}$ we have
    $$\bra{\psi_{11}} \Delta^2 \ket{\psi_{11}} = 5 + 4\bra{\psi_{11}} M_2 \ket{\psi_{11}} \geq 5 - 4|\bra{\psi_{11}} M_2 \ket{\psi_{11}}|\geq 1.$$
    
    Putting everything together, since $\bra{\psi_{\alpha\beta}}\Delta^2 \ket{\psi_{\alpha\beta}} \geq 1$ on each $\mathfrak{L}_{\alpha\beta}$, for $\alpha,\beta \in \{0,1\}$, summing over bases of the respective spaces
    $$\frac{l}{d} = \frac{1}{d} (l_{00} + l_{01} + l_{10} + l_{11}) \leq \frac{1}{d}\sum_{j=1}^{l} \bra{\psi_j} \Delta^2 \ket{\psi_j} \leq \frac{1}{d}\tr(\Delta^2) \leq 8\lambda.$$
    where the second-to-last inequality follows from $\Delta^2$ being positive semidefinite.

    This immediately provides the claimed bound on the statistical difference from uniform. We can explicitly bound the quantities $|a_0|$ and $|a_1|$ as follows:
    \begin{align*}
        |a_0| &= \frac{1}{d}\left|\tr(M_0)\right| = \frac{1}{d} |-l_{00} - l_{01} + l_{10} + l_{11}| \leq \frac{l}{d} \leq 8\lambda\\
        |a_1| &= \frac{1}{d}\left|\tr(M_1)\right| = \frac{1}{d} |-l_{00} + l_{01} - l_{10} + l_{11}| \leq \frac{l}{d} \leq 8\lambda.
    \end{align*}
    Using Cauchy-Schwarz, we bound $|a_2|$:
    \begin{align}
        a_0 + a_1 + a_2 &= \frac{1}{d}\tr(\Delta) \leq \left(\frac{1}{d}\tr(\Delta^2)\right)^{\frac{1}{2}}\left(\frac{1}{d}\tr(\mathbb{1}^2)\right)^{\frac{1}{2}} \leq \sqrt{8\lambda}, \notag \\
        \therefore a_2 &\leq \sqrt{8\lambda} - a_0 - a_1, \notag \\
        \therefore |a_2| &\leq \sqrt{8\lambda} + |a_0| + |a_1| \leq \sqrt{8}\sqrt{\lambda} + 16\lambda. \label{eqn:a2-bound}
    \end{align}
    
   Finally we bound each of $\frac{1}{d}\tr\left(\left(M_x - \tilde{M}_x\right)^2\right)$ in turn. Note $M_0 - \tilde{M}_0 = 0$ by construction. We have
    \begin{align*}
        \frac{1}{d}\tr\left(\left(M_1 - \tilde{M}_1\right)^2\right) &= \frac{1}{d}\sum_j \tr\left(\left(\begin{array}{cc} \cos 2\theta_j - \cos 2\hat\theta & \sin 2\theta_j - \sin 2\hat\theta \\ \sin 2 \theta_j - \sin2\hat\theta  & -\cos 2\theta_j + \cos2\hat\theta\end{array}\right)^2\right)\\
        &=\frac{1}{d} \sum_j (4 - 4\cos(2(\theta_j-\hat\theta))) \ = \ \frac{8}{d} \sum_j \sin^2(\theta_j - \hat\theta)
    \end{align*}
    To bound this, we note that on any $\h{H}_j$:
    $$\left(\begin{array}{cc} 1 & 0 \\ 0 & -1\end{array}\right)\left(\begin{array}{cc} \cos 2\theta_j & \sin 2\theta_j \\ \sin2\theta_j  & -\cos2\theta_j\end{array}\right) + \left(\begin{array}{cc} \cos 2\theta_j & \sin 2\theta_j \\ \sin2\theta_j  & -\cos2\theta_j\end{array}\right)\left(\begin{array}{cc} 1 & 0 \\ 0 & -1\end{array}\right) = 2\cos2\theta_j\cdot \mathbb{1}_{\h{H}_j}.$$
    From this we obtain
    $$\left[\left(\begin{array}{cc} 1 & 0 \\ 0 & -1\end{array}\right) + \left(\begin{array}{cc} \cos 2\theta_j & \sin 2\theta_j \\ \sin2\theta_j  & -\cos2\theta_j\end{array}\right)\right]^2 = 4\cos^2 \theta_j \mathbb{1}_{\h{H}_j}.$$
    Hence there exists a basis $\{\ket{\psi_0},\ket{\psi_1}\}$ of $\h{H}_j$ such that
    \begin{align*}
        (M_0 + M_1) \ket{\psi_0} &= 2\cos\theta_j \ket{\psi_0}\\
        (M_0 + M_1) \ket{\psi_1} &= -2\cos\theta_j \ket{\psi_1}.
    \end{align*}
    Therefore again from (\ref{eqn:Delta-squared}) we have
    \begin{align*}
        \bra{\psi_0} \Delta^2 \ket{\psi_0} &= 3 + 2\cos 2\theta_j + 4\cos\theta_j \bra{\psi_0} M_2 \ket{\psi_0}\\
        \bra{\psi_1} \Delta^2 \ket{\psi_1} &= 3 + 2\cos 2\theta_j - 4\cos\theta_j \bra{\psi_1} M_2 \ket{\psi_1}.
    \end{align*}
    In particular,
    $$\bra{\psi_0} \Delta^2 \ket{\psi_0} +  \bra{\psi_1} \Delta^2 \ket{\psi_1} \geq 6 + 4\cos 2\theta_j - 8|\cos{\theta_j}|.$$
    It is straightforward to show for $\theta \in \left[\frac{2\pi}{3} - \frac{\pi}{6}, \frac{2\pi}{3} + \frac{\pi}{6}\right]$ we have
    $$6 + 4\cos 2\theta - 8|\cos{\theta}| \geq 4\sin^2\left(\theta - \frac{2\pi}{3}\right).$$
    And hence we obtain the bound
    \begin{align*}
        \frac{1}{d}\tr(\Delta^2) &\geq \frac{1}{d}\sum_j (6 + 4\cos 2\theta_j - 8|\cos{\theta_j}|)\\
        &\geq \frac{1}{d}\sum_j 4\sin^2(\theta_j - \hat\theta) \ = \ \frac{1}{2d}\tr\left(\left(M_1 - \tilde{M}_1\right)^2\right).
    \end{align*}
    In particular, $\frac{1}{d}\tr\left(\left(M_1 - \tilde{M}_1\right)^2\right) \leq 16\lambda$.
 
    Finally, note $\tilde{M}_0 + \tilde{M}_1 + \tilde{M}_2 = -\mathbb{1}_{\h{L}_{00}} \oplus \mathbb{1}_{\h{L}_{01}} \oplus -\mathbb{1}_{\h{L}_{10}} \oplus \mathbb{1}_{\h{L}_{11}}$, and hence by Jensen's inequality
    \begin{align*}
        \frac{1}{d}\tr\left((M_2 - \tilde{M}_2)^2\right) &= \frac{1}{d}\tr\left((\Delta - (-\mathbb{1}_{\h{L}_{00}} \oplus \mathbb{1}_{\h{L}_{01}} \oplus -\mathbb{1}_{\h{L}_{10}} \oplus \mathbb{1}_{\h{L}_{11}}) + (\tilde{M_1} - M_1))^2\right)\\
        &\leq \frac{1}{d}\tr\left(\Delta^2\right) + \frac{1}{d}\tr\left(\mathbb{1}_{\h{L}}\right) + \frac{1}{d}\tr\left(\tilde{M_1} - M_1))^2\right)\\
        &\leq 32\lambda.
    \end{align*}

    Therefore, $\frac{1}{3} \sum_{x,y} \frac{1}{d}\tr\left(\left(E^x_y - \tilde{E}^x_y\right)^2\right) \leq 8\lambda$ as desired.
\end{proof}

 It is straightforward to the bound on the statistical difference to any synchronous quantum correlation close to $J_3 = -\frac{1}{8}$. We have every synchronous quantum correlation is a convex sum of synchronous quantum correlations with maximally entangled states (see Appendix \ref{appendix:synchronous}, Lemma 1). So we may write $p = \sum_j c_j p_j$ where $p_j$ is as in the theorem above. Say $J_3(p_j) \leq -\frac{1}{8} + \lambda_j$, and so $$J_3(p) = \sum_j c_j J_3(p_j) \leq -\frac{1}{8} + \sum_j c_j\lambda_j = -\frac{1}{8} + \lambda$$ where we define $\lambda = \sum_j c_j\lambda_j$. With two uses of Jensen's inequality,
    \begin{align*}
        \frac{1}{3}\sum_{x,y}\left| p(y,y\:|\:x,x) - \frac{1}{2}\right| &\leq \frac{1}{3}\sum_{j,x,y} c_j \left|p_j(y,y\:|\:x,x) - \frac{1}{2}\right|\\
        &\leq \sum_j c_j(C\sqrt{\lambda_j} + C'\lambda_j)\\
        &\leq C\sqrt{\lambda} + C'\lambda.
    \end{align*}
    
Unfortunately, this does not yet produce a fully device-independent protocol as we still suffer from a ``synchronicity'' loophole. If a device produces an entangled pair, and Alice and Bob make measurements according to (\ref{eqn:projections-exy}), they can achieve statistics according to (\ref{eqn:protocol-correlation}) if and only if that state is maximally entangled. The self-testing (or rigidity) property requires the protocol to be synchronous, and therefore the loophole is that there may be asynchronous protocols that can produce $J_3 = -\frac{1}{8}$ without using maximally entangled states. We close this loophole in the next section.

\end{section}

\begin{section}{Measure of asynchronicity}\label{section:measures}

That $J_3 = -\frac{1}{8}$ can be achieved by a unique synchronous quantum correlation, which necessarily can only be realized though a maximally entangled state, provides the device-independent security of the above QKD scheme. However this opens a ``synchronicity'' security loophole: can a (non-synchronous) quantum device simulate $J_3 = -\frac{1}{8}$ without using maximally entangled states (and hence potentially leak information about the derived shared keys)? Fortunately a recent work shows that the same results apply to ``almost'' synchronous correlations \cite{vidick2021almost}. This allows us to close this synchronicity loophole by also bounding the asynchronicity of the observed correlation.

\begin{definition}
The \emph{asynchronicity} of a correlation with respect to a basis choice $x\in X$ and set of measurement outcomes $Y$ is
$$S_x(p) = \sum_{y_A\not=y_B} p(y_A,y_B\:|\:x,x).$$
The \emph{total} (or \emph{expected}) asynchronicity is 
\begin{equation}
S(p) = \frac{1}{|X|} \sum_{x\in X} S_x(p) \label{eqn:total-asynchronicity}
\end{equation}
\end{definition}

In \cite{vidick2021almost}, this measure is called the ``default to synchronicity'' and denoted $\delta_{sync}$. As stated, the expected asynchronicity is the average likelihood of a non-synchronous result where the inputs are sampled uniformly at random. All results here, and in \cite{vidick2021almost}, apply to the expected asynchronicity where the expectation is computed over inputs sampled with respect to some other fixed distribution.

In order to statistically bound the asynchronicity, we modify Protocol A so that for some data rounds where Alice and Bob have selected the same inputs they still reveal their output. This results in our Protocol B, found as Algorithm \ref{alg:synchronousQKD-protocolB} below.

Here we state the main result \cite[Theorem 3.1]{vidick2021almost} in the notation used above. Note that this theorem refers to symmetric (albeit nonsynchronous) correlations, which is the natural setting as every synchronous quantum correlation is symmetric. This implies a special form for the projections in the correlation, involving the transpose with respect to the natural basis given by the Schmidt-decomposition of the entangled state used in the correlation.

\begin{theorem}[Vidick]
There are universal constants $c,C > 0$ such that the following holds. Let $X$ and $Y$ be finite sets and $p$ a symmetric quantum correlation with input set $X$, measurement results $Y$, and asynchronicity $S = S(p)$. Write
$$p(y_A,y_B \:|\: x_A,x_B) = \bra\psi E^{x_A}_{y_A}\otimes (E^{x_B}_{y_B})^T \ket\psi$$
where $\{E^x_y\}_{y\in Y}$ is a POVM on a finite-dimensional Hilbert space $\h{H}$ and $\ket\psi$ a state on $\h{H}\otimes\h{H}$. Let
$$\ket\psi = \sum_{j=1}^r \sqrt{\sigma_j} \sum_{m=1}^{d_j} \ket{\phi_{j,m}^A}\otimes \ket{\phi_{j,m}^B}$$
be the Schmidt decomposition, and write $\ket{\psi_j} = \frac{1}{\sqrt{d_j}} \sum_{m=1}^{d_j} \ket{\phi_{j,m}^A}\otimes \ket{\phi_{j,m}^B}$. Then
\begin{enumerate}
    \item $\h{H} = \bigoplus_{j=1}^r \h{H}_j$ with $\ket{\psi_j}$ being maximally entangled on $\h{H}_j\otimes \h{H}_j$;
    \item there is a projective measurement $\{E^{j,x}_y\}_{y\in Y}$ on each $\h{H}_j$ so that
    $$p_j(y_A,y_B\:|\: x_A, x_B) = \bra{\psi_j}E^{j,x_A}_{y_A}\otimes (E^{j,x_B}_{y_B})^T\ket{\psi_j} = \frac{1}{d_j}\tr(E^{j,x_A}_{y_A} E^{j,x_B}_{y_B})$$
    is a synchronous quantum correlation and $p \approx \sum_{j=1}^r d_j\sigma_j p_j$ in that:
    $$\frac{1}{|X|}\sum_{x\in X}\sum_{y\in Y}\sum_{j=1}^r  \frac{1}{d_j} \sum_{m=1}^{d_j} \bra{\phi_{j,m}^A} \left(E^x_y - E^{j,x}_y\right)^2 \ket{\phi_{j,m}^A} \leq C S^c.$$
\end{enumerate}
\end{theorem}

As indicated in \cite[\S{4.1}]{vidick2021almost}, this result can be used to transfer rigidity from synchronous to almost synchronous correlations. As $\sum_{j} d_j\sigma_j = 1$, we and transfer the bound on the statistical difference from uniform to convex sums in this theorem exactly as in the previous section. As for the full correlation we rephrase Lemma 2.10 of \cite{vidick2021almost} in the context of the Theorem as follows.

\begin{corollary}[Vidick]
    Let $p(y_A,y_B \:|\: x_A,x_B) = \bra\psi E^{x_A}_{y_A}\otimes (E^{x_B}_{y_B})^T \ket\psi$ be a quantum correlation with asynchronocity $S$ as in the Theorem, and let $\bar{p} = \sum_{j=1}^r d_j\sigma_j p_j$ with
    $$\frac{1}{|X|}\sum_{x\in X}\sum_{y\in Y}\sum_{j=1}^r  \frac{1}{d_j} \sum_{m=1}^{d_j} \bra{\phi_{j,m}^A} \left(E^x_y - E^{j,x}_y\right)^2 \ket{\phi_{j,m}^A} = \gamma$$
    as given in the Theorem. Then
    $$\frac{1}{|X|^2}\sum_{x_A,x_B,y_A,y_B} | p(y_A,y_B \:|\: x_A,x_B) - \bar{p}(y_A,y_B \:|\: x_A,x_B) | \leq 3S + 4\sqrt{\gamma}.$$
\end{corollary}

Note that this bound on the statistical difference directly bounds $J_3(p)$ in terms of the convex sum of the analogous $J_3(p_j)$. Note that $J_3$, as seen in (\ref{eqn:effective-J3}), is an affine function so $J_3(\bar{p}) = \sum_{j=1}^r \sigma_jd_j J_3(p_j)$ using the notation of the Theorem above. Then immediately from the Corollary, $|J_3(p) - J_3(\bar{p})| \leq \frac{27}{4}S + 9\sqrt{\gamma}$. In turn from the Theorem $\gamma \leq CS^c$, and so there are different universal constants $C',c'$ so that
\begin{equation}\label{eqn:J3-bound}
    |J_3(p) - J_3(\bar{p})| \leq C'S^{c'}.
\end{equation}

\begin{corollary}
    Let $p(y_A,y_B \:|\: x_A,x_B) = \bra\psi E^{x_A}_{y_A}\otimes (E^{x_B}_{y_B})^T \ket\psi$ be a quantum correlation as in the Theorem and suppose $J_3(p) = -\frac{1}{8} + \lambda$. Then the Hilbert space decomposes as $\mathfrak{H} = \bigoplus_{j=1}^r \mathfrak{H}_j = \bigoplus_{j=1}^r (\mathfrak{L}_j \oplus (\mathbb{C}^2 \otimes \mathfrak{K}_j))$ where $\frac{\mathrm{dim}\mathfrak{L_j}}{\mathrm{dim}\mathfrak{H_j}} \leq 8\lambda_j$. On each summand we have projection-valued measures $\{\tilde{E}^{j,x}_y\}$ such that $\tilde{E}^{j,x}_y = L^{j,x}_y + \hat{E}^x_y \otimes \mathbb{1}_{\mathfrak{K}_j}$ and
    $$\frac{1}{3}\sum_{x,y} \sum_{j=1}^{r} \sigma_j d_j \left(\frac{1}{d_j}\sum_{m=1}^{d_j}\bra{\phi^A_{j,m}}(E^{x}_y - \tilde{E}^{j,x}_y)^2\ket{\phi^A_{j,m}}\right) \leq C_1S^{c} + C_2\lambda$$
    for universal constants $c,C_1,C_2$.
\end{corollary}
\begin{proof}
    Given $\{E^x_y\}$ as above, we obtain projections $\{E^{j,x}_y\}$ defining synchronous correlations $p_j$ from the Theorem. Write $J_3(p_j) = -\frac{1}{8} + \lambda_j$. From Theorem \ref{theorem:main_bound}, we obtain the given decomposition of the Hilbert space and projection-valued measures $\{\tilde{E}^{j,x}_y\}$ where 
    \begin{enumerate}
        \item $\tilde{E}^{j,x}_y = L^{j,x}_y + \hat{E}^x_j\otimes \mathbb{1}_{\mathfrak{K}_j}$,
        \item $\frac{\dim{\mathfrak{L}_j}}{\dim{\mathfrak{H}_j}} \leq 8 \lambda_j$, and
        \item $\frac{1}{3}\sum_{x,y}\frac{1}{d_j}\sum_{m=1}^{d_j}\bra{\phi^A_{j,m}}(E^{j,x}_y - \tilde{E}^{j,x}_y)^2\ket{\phi^A_{j,m}} \leq C_2\lambda_j$.
    \end{enumerate}  
    Then using the notation and (\ref{eqn:J3-bound}) above $|J_3(p) - J_3(\bar{p})| = \left|\lambda - \sum_{j=1}^r \sigma_j d_j \lambda_j\right| \leq C'S^{c'}$ and thus
    \begin{align*}
        &\frac{1}{3}\sum_{x,y} \sum_{j=1}^{r}\sigma_j d_j \left(\frac{1}{d_j}\sum_{m=1}^{d_j}\bra{\phi^A_{j,m}}(E^{j,x}_y - \tilde{E}^{j,x}_y)^2\ket{\phi^A_{j,m}}\right)\\
        &\qquad \leq C_2 \sum_{j=1}^r\sigma_j d_j \lambda_j = C_2\lambda + C_2C'S^{c'}.
    \end{align*}
    
    On the other hand,
    \begin{align*}
        &\frac{1}{3}\sum_{x,y} \sum_{j=1}^{r} \sigma_j d_j \left(\frac{1}{d_j}\sum_{m=1}^{d_j}\bra{\phi^A_{j,m}}(E^{x}_y - E^{j,x}_y)^2\ket{\phi^A_{j,m}}\right)\\
        &\qquad \leq \frac{1}{3}\sum_{x,y} \sum_{j=1}^{r} \left(\frac{1}{d_j}\sum_{m=1}^{d_j}\bra{\phi^A_{j,m}}(E^{x}_y - E^{j,x}_y)^2\ket{\phi^A_{j,m}}\right) \leq C'' S^{c''}
    \end{align*}
    directly from the Theorem. So by Jensen's inequality
    \begin{align*}
        &\frac{1}{3}\sum_{x,y} \sum_{j=1}^{r} \sigma_j d_j \left(\frac{1}{d_j}\sum_{m=1}^{d_j}\bra{\phi^A_{j,m}}(E^{x}_y - \tilde{E}^{j,x}_y)^2\ket{\phi^A_{j,m}}\right)\\
        &\qquad \leq \frac{2}{3}\sum_{x,y} \sum_{j=1}^{r} \sigma_j d_j \left(\frac{1}{d_j}\sum_{m=1}^{d_j}\bra{\phi^A_{j,m}}(E^{x}_y - E^{j,x}_y)^2\ket{\phi^A_{j,m}}\right)\\
        &\qquad\qquad +\ \frac{2}{3}\sum_{x,y} \sum_{j=1}^{r}\sigma_j d_j \left(\frac{1}{d_j}\sum_{m=1}^{d_j}\bra{\phi^A_{j,m}}(E^{j,x}_y - \tilde{E}^{j,x}_y)^2\ket{\phi^A_{j,m}}\right)\\
        &\qquad\leq 2C_1S^c + 2C_2\lambda
    \end{align*}
    for some universal constant $C_1$.
\end{proof}

\begin{table}[t]
    \begin{algorithm}[H]\label{alg:synchronousQKD-protocolB}
        \caption{Protocol B}
        \SetAlgoLined
        \DontPrintSemicolon
        \SetKwFor{For}{for}{do}{}
        \KwIn{$\lambda, \mu, n, m$}
        $X \leftarrow \{0,1,2\}$ and $Y \leftarrow \{0,1\}$\\
        Alice and Bob share $n$ EPR pairs: $\ket{\psi} = \frac{1}{\sqrt{2}}\left(\ket{00} + \ket{11}\right)$\\
        Alice and Bob both have available three particular measurement bases $\{E^x_y\}_{x\in X, y\in Y}$\\
        \For{$i = 1, \cdots, n$}{
            Alice draws $x^i_A \overset{\$}{\leftarrow} X$ and Bob draws $x^i_B \overset{\$}{\leftarrow} X$\\
            With the $i^\text{th}$ EPR pair, Alice obtains $y^i_A$ using  $\{E^{x^i_A}_y\}$ and Bob obtains $y^i_B$ (using $\{E^{x^i_B}_y\}$)
        }
        Alice and Bob exchange their choices of $x^i_A, x^i_B$, for $i \in [n]$\\
        Whenever $x^i_A \neq x^i_B$, or when $x^i_A = x^i_B$ and $i=0 \pmod{m}$, Alice and Bob exchange $y^i_A, y^i_B$\\
        $k \leftarrow \emptyset$\\
        \For {$i = 1,\cdots, n$}{
            \uIf{$x^i_A = x^i_B$}{
                \uIf{$i = 0 \pmod{m}$}{
                    Add result to estimation of $S$ 
                }
                \lElse{
                    $k \leftarrow k \cup y^i$, where $y^i := y^i_A = y^i_B$ due to synchronicity
                }
            }
            \lElse{
                Add result to estimation of $J_3$
            }            
        }
        Compute an estimate $\hat{J}_3$ using (\ref{eqn:effective-J3})\\
        Compute and estimate $\hat{S}$ using (\ref{eqn:total-asynchronicity})\\
        \uIf{ $|\hat{J}_3 + \frac{1}{8}| \leq \lambda$ and $S \leq \mu$}{
            \textbf{Return } $k$ (for standard information reconciliation and privacy amplification)
        }\lElse{ Abort }
    \end{algorithm}
    \end{table}

\end{section}

\begin{section}{Causality Loophole}\label{section:causality-loophole}
    In this section we describe what is called the \emph{causality} or \emph{locality} loophole common to device independent quantum key distribution protocols that use non-local games, and propose a solution to the loophole using a new security assumption.
    
    As seen in the previous section, the bound for the Bell inequality $ J_3 \geq -\frac{1}{8}$ is sharp and rigid only among synchronous quantum correlations. There exist more powerful synchronous non-signaling strategies that violate those bounds. Furthermore, if classical communication is allowed between the parties in the protocol, even greater violations can be achieved. This is the \emph{causality loophole}: unless Alice and Bob are acausally separated, then the statistics for the synchronous Bell inequalities can simply be simulated using classical communication.
    
    In order to resolve the causality loophole in our protocol, we pose a new security assumption: Instead of limiting Eve's computational power or limiting the communication she can perform, we assume that she has imperfect knowledge of the basis Alice and Bob use in the protocol. We state this more formally:
    
    \noindent Let $\epsilon$ be Eve's uncertainty about Alice and Bob's inputs. Without loss of generality, we assume that her uncertainty is symmetric across all basis selections. For $x', x \in \{0,1,2\}$ we have
    
    \[
        \Pr\{\text{Eve guesses basis } x' \given \text{Alice (or Bob) selects basis } x\} = 
            \left\{
                \begin{array}{cl} 
                    1 - \epsilon & \text{ when $x' = x$}\\\\ 
                    \frac{\epsilon}{2} & \text{ when $x' \not= x$.}
                \end{array}
            \right.
    \]
    
    We denote Eve's guess for Alice's input by $z_A$ and for Bob's input by $z_B$. Eve has unlimited computational power and communication and can use any strategy of her choosing to produce outputs $(y_A, y_B)$. We denote her correlation as $\Pr\{(y_A ,y_B \given z_A, z_B)\}$. The correlation that Alice and Bob use to compute key bits and self-test their devices is then given by:
    
    \begin{align*}
        p(y_A,y_B &\given x_A,x_B) = \\
        &\sum_{z_A,z_B} \Pr\{ (y_A,y_B) \given (z_A,z_B)\}\cdot
        \left\{\begin{array}{cl} 1 - \epsilon & \text{for $z_A = x_A$} \\\\ \frac{\epsilon}{2} & \text{otherwise}\end{array}\right\}\cdot
        \left\{\begin{array}{cl} 1 - \epsilon & \text{for $z_B = x_B$} \\\\ \frac{\epsilon}{2} & \text{otherwise}\end{array}\right\}.
    \end{align*}

\begin{theorem}
    Let $0\leq \epsilon \leq \frac{2}{3}$ be Eve's uncertainty. Let $0 \leq \lambda \leq \frac{1}{8}$ and $0 \leq \mu \leq \mu_0$ be allowed errors in expected values for Alice and Bob's Bell term $J_3$ and asynchronicity $S$ respectively. We write analogous terms $\tilde{J}_3$ and $\tilde{S}$ for Eve's strategy. Let
    $$\epsilon_{max} = \frac{2}{3} - \frac{2}{3}\left(\frac{\sqrt{64\lambda^2 + 6(8\lambda - 9)\mu - 72\mu^2 - 144\lambda + 81}}{6\mu - 8\lambda  + 9}\right).$$ If Eve's uncertainty is $\epsilon > \epsilon_{max}$ then every correlation satisfies $\tilde{S} < 0$, and hence there is no feasible strategy she can produce. Said another way, all feasible strategies only exist for $\epsilon \in [0, \epsilon_{max}]$.\\
    
    \noindent Furthermore, if Eve's asynchronicity is bounded below by $\delta$ i.e. $ 0 \leq \delta \leq \tilde{S} \leq \mu$, then the maximum uncertainty she could have before her asynchronicity $S > \delta$ is $$\epsilon^{\delta}_{max} = \frac{2}{3} - \frac{2}{3}\left(\frac{\sqrt{144(\delta - 1)\lambda + 64\lambda^2 + 6(36\delta + 8\lambda - 9)\mu - 72\mu^2 - 162\delta + 81}}{6\mu - 18\delta - 8\lambda + 9}\right)$$ Note that $0 \leq \epsilon^\delta_{max} \leq \epsilon_{max}$.
\end{theorem}
\begin{proof}
    We begin by deriving expressions for the expected values of $J_3$ and $S$.
   \begin{align}
    \langle 1 - J_3 \rangle &= \frac{1}{4}\big( p(0,1\:|\:0,1) + p(1,0\:|\:0,1) + p(0,1\:|\:1,0) + p(1,0\:|\:1,0) \notag\\
        &\quad +\ p(0,1\:|\:0,2) + p(1,0\:|\:0,2) + p(0,1\:|\:2,0) + p(1,0\:|\:2,0) \notag \\
        &\quad +\ p(0,1\:|\:1,2) + p(1,0\:|\:1,2) + p(0,1\:|\:2,1) + p(1,0\:|\:2,1)\big) \notag\\
    &=  \left(1 - \epsilon + \tfrac{3}{4}\epsilon^2\right)(1-\tilde{J}_3) + \left(\tfrac{3}{2}\epsilon - \tfrac{9}{8}\epsilon^2\right)\tilde{S} \label{eqn:expected_oneminusJ3}
    \end{align}
    A similar computation for $S$ gives us:
    \begin{align}
        \langle S \rangle &= \frac{1}{3}\big(p(0,1\:|\:0,0) + p(1,0\:|\:0,0) + p(0,1\:|\:1,1) \notag\\
        & \qquad + p(1,0\:|\:1,1) + p(0,1\:|\:2,2) + p(1,0\:|\:2,2)\big) \notag \\
        &= \left(1 - 2\epsilon + \tfrac{3}{2}\epsilon^2\right)\tilde{S} + \left(\tfrac{4}{3}\epsilon - \epsilon^2\right)(1-\tilde{J}_3) \label{eqn:expectedS}
    \end{align}
    Using (\ref{eqn:expected_oneminusJ3}) and (\ref{eqn:expectedS}), we can solve for $\tilde{J}_3$ and $\tilde{S}$ as:
    \begin{align*}
        \begin{bmatrix}
            1 - \tilde{J}_3 \\
            \tilde{S}
        \end{bmatrix}
        &=
        \begin{bmatrix*}[r]
            1 - \epsilon + \tfrac{3}{4}\epsilon^2 & \tfrac{3}{2}\epsilon - \tfrac{9}{8}\epsilon^2 \vspace{0.2cm} \\
            \tfrac{4}{3}\epsilon - \epsilon^2 & 1 - 2\epsilon + \tfrac{3}{2}\epsilon^2
        \end{bmatrix*}^{-1}
        \begin{bmatrix}
            \frac{9}{8} - \lambda \\
            \mu
        \end{bmatrix}
    \end{align*}
    We get solutions:
    \begin{align}
        \label{eqn:eve's-J3} \tilde{J}_3 &= 1 - \frac{ (3\epsilon^2 - 4\epsilon)(6\mu  - 8\lambda + 9) - 16\lambda + 18}{4\left(3\epsilon - 2\right)^{2}} = \frac{ (3\epsilon^2 - 4\epsilon)(3 - 6\mu  + 8\lambda) + 16\lambda - 2}{4\left(3\epsilon - 2\right)^{2}}\\
        \label{eqn:eve's-S} \tilde{S} &= \frac{(3\epsilon^2 - 4\epsilon)(6\mu - 8\lambda + 9) + 24\mu}{6\left(3\epsilon - 2\right)^{2}}.
    \end{align}

\noindent Plugging $\tilde{S} = \delta$ in (\ref{eqn:eve's-S}), and solving for $\epsilon$ gives us:
\begin{align*}
    \epsilon^{\delta}_{max} = \frac{2}{3} - \frac{2}{3}\left(\frac{\sqrt{144(\delta - 1)\lambda + 64\lambda^2 + 6(36\delta + 8\lambda - 9)\mu - 72\mu^2 - 162\delta + 81}}{6\mu - 18\delta - 8\lambda + 9}\right)
\end{align*}
For $\delta = 0$, we get 
\begin{align*}
    \epsilon^0_{max} = \epsilon_{max} = \frac{2}{3} - \frac{2}{3}\left(\frac{\sqrt{64\lambda^2 + 6(8\lambda - 9)\mu - 72\mu^2 - 144\lambda + 81}}{6\mu - 8\lambda  + 9}\right)
\end{align*}
\end{proof}

\noindent By the theorem above, we conclude that Eve's uncertainty cannot grow too much before her asynchronicity becomes negative, therefore resulting in an infeasible strategy. Fixing $\lambda=1/8$, which is the maximum possible error allowed in the Bell term, we plot values of $\epsilon_{max}$ against varying values of Alice and Bob's allowed asynchronicity $\mu$ in Figure (\ref{fig:asynchronicity_v_epsilon}). We also fix a value for Eve's asynchronicity $\tilde{S} = \delta = 0.01$, and plot the maximum value for her uncertainty against varying values of $\mu$ in Figure (\ref{fig:asynchronicity_v_epsilon_delta}). The first plot shows that even for allowed asynchronicity $\mu=5\%$, Eve must have close to perfect certainty $ \approx 97\%$ about Alice and Bob's inputs, otherwise she cannot simulate the statistics for the protocol despite unlimited computational power. 

\begin{figure}[htb]
\centering
\makebox[0pt][c]{%
\begin{minipage}[b]{.47\textwidth}
    \centering
    \begin{tikzpicture}[y=1.2cm, x=1cm]
        \draw (0,0) -- coordinate (x axis mid) (6,0) node[anchor=north west] {$\mu$};
        \draw (0,0) -- coordinate (y axis mid) (0,5) node[anchor=south east] {$\epsilon_{max}$};
        \draw (0, 1pt) -- (0, -3pt) node[anchor=north] {$0$};
        \foreach \x in {1, 2, 3, 4, 5}
            {
            \draw (\x - 0.5, 1pt) -- (\x - 0.5, -3pt) node[anchor=north] {};
            \draw (\x, 1pt) -- (\x, -3pt) node[anchor=north] {$0.0\x$};
            }
        \foreach \y in {1, 2, 3, 4}
            \draw (1pt, \y) -- (-3pt, \y) node[anchor=east] {$\y\times 10^{-2}$};
    	\draw plot[mark=*, mark options={fill=magenta}]
    	    file {eps_max-mu.data};
    	\begin{scope}[shift={(4,4)}] 
            \draw (-3.5,1.0) node[right]{$\lambda = \frac{1}{8}$};
	    \end{scope}
    \end{tikzpicture}
    \captionof{figure}{Values of $\mu$ vs. $\epsilon_{max}$ for which Eve's asynchronicity $\tilde{S}$ is positive}
    \label{fig:asynchronicity_v_epsilon}
\end{minipage} \qquad
\begin{minipage}[b]{.47\textwidth}
   \centering
    \begin{tikzpicture}[y=1.2cm, x=1cm]
        \draw (0,0) -- coordinate (x axis mid) (6,0) node[anchor=north west] {$\mu$};
        \draw (0,0) -- coordinate (y axis mid) (0,5) node[anchor=south east] {$\epsilon^\delta_{max}$};
        \draw (0, 1pt) -- (0, -3pt) node[anchor=north] {$0$};
        \foreach \x in {1, 2, 3, 4, 5}
            {   
            \draw (\x - 0.5, 1pt) -- (\x - 0.5, -3pt) node[anchor=north] {};
            \draw (\x, 1pt) -- (\x, -3pt) node[anchor=north] {$0.0\x$};
            }
        \foreach \y in {1, 2, 3, 4}
            \draw (1pt, \y) -- (-3pt, \y) node[anchor=east] {$\y\times 10^{-2}$};
    	\draw plot[mark=*, mark options={fill=magenta}]
    	    file {eps_delta_max-mu.data};
    	\begin{scope}[shift={(4,4)}] 
            \draw (-3.5,1.0) node[right]{$\tilde{S} = \delta = 0.01, \lambda = \frac{1}{8}$};
	    \end{scope}
    \end{tikzpicture}
    \captionof{figure}{Values of $\mu$ vs. $\epsilon^\delta_{max}$ for which Eve's asynchronicity $\tilde{S}\geq \delta$}
    \label{fig:asynchronicity_v_epsilon_delta}
\end{minipage}%
}%
\end{figure}

\end{section}

\appendix

\newpage
\begin{section}{Synchronous correlations}\label{appendix:synchronous}

For completeness we present a more detailed discussion of synchronous correlations from \cite{paulsen2016estimating,rodrigues2017nonlocal}. Recall a \emph{local hidden variables strategy}, or simply \emph{classical correlation}, is a correlation of the form
    \begin{equation}\label{eqn:classical:definition-appendix}
        p(y_A,y_B\:|\:x_A,x_B) = \sum_{\omega\in\Omega} \mu(\omega) p_A(y_A\:|\: x_A,\omega) p_B(y_B\:|\: x_B,\omega)
    \end{equation}
for some finite set $\Omega$ and probability distribution $\mu$. For $p$ to by synchronous, if Alice and Bob input the same $x\in X$, then they must produce $y_A = y_B$ with certainty. That is, for every $\omega\in \Omega$ we must have $p_A(y_A\:|\: x,\omega) p_B(y_B\:|\: x,\omega) = 0$
whenever $y_A \not= y_B$. This in turn implies that for each $\omega$ there is a function $f:X \to Y$ so that $p_A(y\:|\: x,\omega) = p_B(y\:|\: x,\omega) = \chi(\{y = f_\omega(x)\})$ where $\chi(\cdot)$ is the indicator function. That is, any classical synchronous correlation is given by the following strategy: Alice and Bob (randomly) pre-select a function $f:X \to Y$, and upon given $x_A,x_B\in X$ each computes their respective outputs $y_A = f(x_A)$ and $y_B = f(x_B)$. Consequently every classical synchronous correlation is also symmetric.

A \emph{quantum correlation} is a correlation that takes the form
    \begin{equation}\label{eqn:quantum:definition}
        p(y_A,y_B\:|\:x_A,x_B) = \tr(\rho(E^{x_A}_{y_A}\otimes F^{x_B}_{y_B}))
    \end{equation}
where $\rho$ is a density operator on the Hilbert space $\h{H}_A\otimes\h{H}_B$, and for each $x\in X$ we have $\{E^x_y\}_{y\in Y}$ and $\{F^x_y\}_{y\in Y}$ are POVMs on $\h{H}_A$ and $\h{H}_B$ respectively. 
We will only treat the case when $\h{H}_A$ and $\h{H}_B$ are finite dimensional.

One generally argues that by enlarging the Hilbert spaces one can take the assumed POVMs in the definition to be projection-valued measures. But for synchronous quantum correlations this must already be true \cite[Proposition 1]{cameron2007quantum}, but see also  \cite{abramsky2017quantum,atserias2016quantum,manvcinska2016quantum}.

The works cited above a common result is that if a synchronous quantum correlation exists that satisfies some additional properties, then another such correlation exists whose state is maximally entangled; examples of such include \cite[Proposition 1]{cameron2007quantum}, \cite[Lemma 4]{abramsky2017quantum}, \cite[Theorem 2.1]{manvcinska2016quantum}. It is certainly not the case that every synchronous quantum correlation can be taken to have a maximally entangled state, as these include hidden variables strategies. Nonetheless we can prove that every synchronous quantum correlation is a convex sum of such, and for such correlation we can express it as a so-called ``tracial'' state, as given below.

\begin{lemma}
    Every synchronous quantum correlation can be expressed as the convex combination of synchronous quantum correlations with maximally entangled pure states. In particular, if a synchronous quantum correlation $\tr(\rho(E^{x_A}_{y_A}\otimes F^{x_B}_{y_B}))$ is extremal then we may take $\rho = \ket\psi\bra\psi$ with $\ket\psi$ maximally entangled.
\end{lemma}

\begin{theorem}\label{theorem:tracial-state}
Let $X,Y$ be finite sets, $\h{H}$ a $d$-dimensional Hilbert space, and for each $x\in X$ a projection-valued measure $\{E^x_y\}_{y\in Y}$ on $\h{H}$. Then
$$p(y_A,y_B\:|\: x_A,x_B) = \frac{1}{d}\tr(E^{x_A}_{y_A}E^{x_B}_{y_B})$$
defines a synchronous quantum correlation. Moreover every synchronous quantum correlation with maximally entangled pure state has this form.
\end{theorem}

\begin{corollary}
    Every synchronous quantum correlation is symmetric.
\end{corollary}

When studying correlations with $|X| = n$ and $|Y|=2$, and for concreteness say $Y = \{0,1\}$, it is particularly fruitful to work with the traditional biases and correlation matrices:
\begin{align*}
    a_{x_A} &= \sum_{y_A, y_B}{(-1)^{(1,0)\cdot(y_A, y_B)} p(y_A, y_B | x_A, x_B)}\\
    b_{x_B} &= \sum_{y_A, y_B}{(-1)^{(0,1)\cdot(y_A, y_B)} p(y_A, y_B | x_A, x_B)}\\
    c_{x_A,x_B} &= \sum_{y_A, y_B}{(-1)^{(1,1)\cdot(y_A, y_B)} p(y_A, y_B | x_A, x_B)}.
\end{align*}
Note that the nonsignaling criteria implies that $a$ and $b$ do not depend on $x_B$ or $x_A$ respectively. 

The properties of a correlation being symmetric or synchronous can easily expressed in this variables. As indicated this form exists only for nonsignaling correlations, and so that must be included in the characterization.

\begin{proposition}
    The following hold:
    \begin{enumerate}
        \item A correlation $p$ is symmetric and nonsignaling if and only if (i) $c_{x_A,x_B} = c_{x_B,x_A}$ and (ii) $a_x = b_x$.
        \item A correlation $p$ is synchronous and nonsignaling if and only if for all $x\in X$ we have (i) $c_{x,x} = 1$ and (ii) $a_x = b_x$.
        \item A correlation matrix $C$ is synchronous and quantum if and only if there exists unit vectors $\{\vec{u}_x\}$ such that $c_{x_A,x_B} = \langle \vec{u}_{x_A}, \vec{u}_{x_B}\rangle$.
    \end{enumerate}
\end{proposition}

We see that a general symmetric, synchronous, nonsignaling correlation has $c_{x_A,x_B} = c_{x_B,x_A}$, $c_{x,x} = 1$, and $a_x = b_x$. The set of such correlation forms a polytope. The classical synchronous correlations form a subpolytope of this set, and the inequalities from facets of this subpolytope that are not already facets of the larger set define \emph{synchronous Bell inequalities}.

For even moderate size $X$ it is complicated to find all such Bell inequalities, but for small $X$ this is tractable. For example, if $X = \{0,1\}$ then there are no Bell inequalities: every symmetric, synchronous, nonsignaling correlation (and hence also synchronous quantum correlation) is classical. The case of interest for this work is $X = \{0,1,2\}$ for which the synchronous Bell inequalities are:
\begin{equation}\label{eqn:Bell-inequalities}
    \begin{array}{rll}
    J_0 &= \tfrac{1}{4}\left(1 - c_{01} - c_{02} + c_{12}\right) &\geq 0\\
    J_1 &= \tfrac{1}{4}\left(1 - c_{01} + c_{02} - c_{12}\right) &\geq 0\\
    J_2 &= \tfrac{1}{4}\left(1 + c_{01} - c_{02} - c_{12}\right) &\geq 0\\
    J_3 &= \tfrac{1}{4}\left(1 + c_{01} + c_{02} + c_{12}\right) &\geq 0.
    \end{array}
\end{equation}
So a symmetric synchronous nonsignaling correlation, which includes any synchronous quantum correlation, is classical if and only if these four inequalities are satisfied.

Quantum correlations can violate the Bell inequalities (\ref{eqn:Bell-inequalities}). Yet, we can show there are maximal quantum violations akin to Tsirl'son bounds on Bell's inequality \cite{tsirelson1980quantum}. Note that we have stated this as in \cite[Theorem 15]{rodrigues2017nonlocal}, however the core of proof of this is already present our equation (\ref{eqn:trace-Delta-squared-bound}) above: $\frac{1}{d}\tr(\Delta^2) = 1 + 8J_3$ is nonnegative.

\begin{theorem}\label{theorem:synchronous-quantum-bounds}
    Every synchronous quantum correlation satisfies $J_0, J_1, J_2, J_3 \geq -\frac{1}{8}$. However no individual correlation can violate more than one of the inequalities $J_0, J_1, J_2, J_3 \geq 0$.
\end{theorem}

In CHSH, and similar nonlocal games, device-independence is a consequence of the rigidity of quantum correlations that achieve a maximal quantum violation. Identical rigidity results are true of the four synchronous quantum correlations above.

\begin{theorem}\label{theorem:rigidity}
    For each of the four bounds of Theorem \ref{theorem:synchronous-quantum-bounds}, there exists a unique synchronous quantum correlation from $\{0,1,2\}$ to $\{0,1\}$ that achieves it.
\end{theorem}

For example, the unique correlation (among synchronous quantum correlations) with $J_3 = -\frac{1}{8}$ has $a_0 = a_1 = a_2 = 0$ and $c_{0,1} = c_{0,2} = c_{1,2} = -\frac{1}{2}$. One can then easily convert this expression into (\ref{eqn:effective-J3}).

\end{section}

\end{document}